\documentclass[11pt,reqno]{amsart}
\usepackage{graphicx,amsmath,amsthm,verbatim,amssymb,lineno,eucal,titletoc} %eucal is "Euler" mathcal for the script N denoting an abelian network
\usepackage{bbm}
\usepackage{hyperref}
\usepackage{complexity}
\hypersetup{colorlinks}
\pdfoutput=1

\newcommand{\arxiv}[1]{{\tt \href{http://arxiv.org/abs/#1}{arXiv:#1}}}

\newcommand{\Set}[2]{\left\{ #1 \, \left| \; #2 \right. \right\}}

\newcommand{\ghost}[1]{\textcolor{white}{#1}}

\newcommand{\old}[1]{}
\newcommand{\moniker}[1]{{\em (#1)}}

\DeclareRobustCommand{\SkipTocEntry}[5]{}
\newcommand{\silentsec}[1]{\addtocontents{toc}{\SkipTocEntry}\section*{#1}}  % won't appear in TOC
\newcommand{\silentsubsec}[1]{\addtocontents{toc}{\SkipTocEntry}\subsection*{#1}}

\newtheorem{theorem}{Theorem}[section]

\newtheorem{lemma}[theorem]{Lemma}
\newtheorem{corollary}[theorem]{Corollary}

\theoremstyle{remark}
\newtheorem*{example}{Example}
\newtheorem*{remark}{Remark}

\numberwithin{counter}{section}

\theoremstyle{definition}
\newtheorem{definition}[theorem]{Definition}

\def\K{K} %sandpile group
\def\isom{\simeq}

\def\mm{\mathbf{m}}

\def\qq{\mathbf{q}}
\def\rr{\mathbf{r}}
\def\ss{\mathbf{s}}

\def\uu{\mathbf{u}}
\def\vv{\mathbf{v}}
\def\ww{\mathbf{w}}
\def\xx{\mathbf{x}}
\def\yy{\mathbf{y}}
\def\zz{\mathbf{z}}
\def\mm{\mathbf{m}}
\def\kk{\mathbf{k}}

\def\Proc{\CMcal{P}}  % regular mathcal
\def\Net{{\mathcal{N}}}  % "Euler" mathcal

\def\Crit{\mathrm{Crit\,}}
\def\acts{\mathop{\triangleright}}

\def\End{\mathrm{End\, }}
\def\zero{\mathbf{0}}
\def\one{\mathbf{1}}
\def\Sa{\mathcal{S}}

\def\basis{1}
\def\Rotor{{\tt Rotor}}
\def\Sand{{\tt Sand}}
\def\Topp{{\tt Topp}}
\def\N{\mathbb{N}}
\def\Z{\mathbb{Z}}
\def\Q{\mathbb{Q}}
\def\R{\mathbb{R}}

\def\eps{\epsilon}

\begin{document}

\title[Abelian Networks]{Abelian Networks II. Halting On All Inputs}

\author[Benjamin Bond and Lionel Levine]{Benjamin Bond and Lionel Levine}

\address{Benjamin Bond, Department of Mathematics, Stanford University, Stanford, California 94305. {\tt\url{http://stanford.edu/~benbond}}}

\address{Lionel Levine, Department of Mathematics, Cornell University, Ithaca, NY 14853. {\tt \url{http://www.math.cornell.edu/~levine}}}

%\date{submitted August 30, 2014}
\date{August 4, 2015}

\begin{abstract}
Abelian networks are systems of communicating automata satisfying a local commutativity condition. We show that a finite irreducible abelian network halts on all inputs if and only if all eigenvalues of its production matrix lie in the open unit disk. 
\end{abstract}

\maketitle
%\newpage
%\tableofcontents
%
\section{Introduction}

Automata networks in general are nondeterministic: the same input can produce many different outputs depending on the order of events at different nodes of the network.
However, there is a sizable class of automata networks, proposed by Dhar in \cite{Dha99} and termed \emph{abelian networks} in \cite{part1}, for which the output is uniquely determined by the input.  
%In an abelian network not only the output but also the total run time and even the load on each individual processor are all independent of the execution order.

If we take the view that an abelian network is a kind of asynchronous computer, then one of the most fundamental questions is its \emph{halting problem}.  A countably infinite abelian network can emulate a Turing machine with infinite tape \cite{Hannah}, so the question of whether a given abelian network halts on a given input is undecidable.  However, for a \emph{finite} abelian network it is decidable (below we give an argument using Dickson's lemma) and its computational complexity is an interesting question.  

The first halting criterion in an abelian network that we know of is due to Tardos \cite{Tar88}, who found an efficient check for finite termination of chip-firing on an undirected graph. Chip-firing belongs to a subclass of abelian networks called \emph{toppling networks}, which can be described informally as follows.  Let $L$ be an integer square matrix with positive diagonal entries nonpositive off-diagonal entries. Each vertex $v$ has a number of \emph{chips} $c_v$ and is allowed to \emph{topple} if $c_v \geq L_{vv}$; the result of toppling $v$ is that $v$ loses $L_{vv}$ chips and each other vertex $u \neq v$ gains $-L_{uv}$ chips. Halting means that there exists a finite sequence of topplings after which $c_v < L_{vv}$ for all $v$.

The Tardos halting criterion applies when $L$ is the Laplacian of an undirected (or Eulerian directed) graph. The Laplacian of a general directed graph is more difficult because there exist directed graphs with chip configurations $c$ that require exponentially many topplings to halt.  Bj\"{o}rner, Lov\'{a}sz and Shor \cite{BLS91} remarked, ``\emph{one can ask for a characterization of those digraphs and initial chip configurations that guarantee finite termination}.'' This is a difficult problem: Bj\"{o}rner and Lov\'{a}sz \cite{BL92} gave an algorithm that takes exponential time in the worst case. Polynomial time algorithms are available for Eulerian and coEulerian graphs, but the problem is $\NP$-complete for a general directed multigraph \cite{FL15}.

The halting problem for abelian networks is at least as hard as the problem posed by Bj\"{o}rner, Lov\'{a}sz and Shor.
In this paper we address the softer question of characterizing which finite abelian networks $\Net$ halt on \emph{all} inputs. We associate to $\Net$ a \emph{production matrix} $P$ whose entries are nonnegative rational numbers, and show that $\Net$ halts on all inputs if and only if the Perron-Frobenius eigenvalue of $P$ is strictly less than $1$. This result generalizes an unpublished theorem of Gabrielov \cite{Gab94}, who considered the case of a toppling network. 

\subsection*{Outline}
Section~\ref{s.background} contains mathematical background on commutative monoid actions, Dickson's lemma and toppling matrices. Theorem~\ref{t.monoidtogroup} is of some independent interest: it gives a general mechanism for how torsors (free transitive actions) of abelian groups arise from monoid actions. 
%Our main application of this theorem comes in \cite{part3} when we will apply it to the critical group of an abelian network.

After reviewing the definition of an abelian network in Section~\ref{s.networks}, we define two basic algebraic objects associated to an abelian network, the \emph{total kernel} and \emph{production matrix}, in Section~\ref{s.KP}.
We then examine two ways of decomposing an abelian network: local components and strong components. The former have a smaller state space, the latter a smaller alphabet.

In Section~\ref{s.halting} we prove our criterion for halting on all inputs.

In Section~\ref{s.laplacian} we define the \emph{sandpilization} $\Sa(\Net)$ of an abelian network $\Net$, and show that $\Net$ halts on all inputs if and only if $\Sa(\Net)$ does.

We conclude in Section~\ref{s.concluding} with a short discussion of further research directions.

\section{Mathematical background}
\label{s.background}

%This section reviews mathematical background on monoid actions, Dickson's lemma and toppling matrices. 

\subsection{Commutative monoid actions}
\label{s.monoid}

Any finite commutative monoid $M$ contains an abelian group whose identity element is the minimal idempotent of $M$.  Every monoid action of $M$ induces a corresponding group action.  The main result of this section is Theorem~\ref{t.monoidtogroup} relating these two actions.
We have not seen this theorem stated explicitly in the literature, but some of the lemmas in this section are well-known in the semigroup community.  They trace their origins to the work of Green \cite{Gre51} and Sch\"{u}tzenberger \cite{Sch57}; see \cite{Gri01,Ste10} for modern treatments.  We include their short proofs here in order to highlight the beauty and simplicity of the commutative case.  For refinements of some of the lemmas below, and extensions to a certain class of infinite semigroups ($\pi$-regular semigroups) see \cite{Gri07}.

Let $M$ be a commutative monoid, that is, a set equipped with a commutative and associative operation $(m,m') \mapsto mm'$ with an identity element $\eps \in M$ satisfying $\eps m = m$ for all $m \in M$. Let
	\begin{align*} \mu: M \times X &\to X \\ (m,x) &\mapsto mx \end{align*}
be a monoid action of $M$ on a set $X$; that is, $\eps x=x$ and $m(m'x) = (mm')x$ for all $m,m' \in M$ and all $x \in X$.  
% we don't require X finite, but irreducible will force this if M is finite.
We say that $\mu$ is \emph{irreducible} if there does not exist a partition of $X$ into nonempty subsets $X_1$ and $X_2$ such that $MX_1 \subset X_1$ and $MX_2 \subset X_2$.  Consider the relation $\sim$ on $X$ defined by
	\begin{equation} \label{e.access}  x \sim x' : \qquad \exists m,m' \in M \text{ such that } mx=m'x' \end{equation}

\begin{lemma}
\label{l.chainshortening}
$\sim$ is an equivalence relation, and $\mu$ is irreducible if and only if $\sim$ has just one equivalence class.
%$\mu$ is irreducible if and only if for all $x,x' \in X$ there exist $m,m' \in M$ such that $mx=m'x'$.
\end{lemma}

% in fact, something stronger is true: there exists an x that is accessible from every x' (namely, x can be anything in eX).  But this seems harder to prove from scratch.

\begin{proof}
Clearly $x \sim x$, and $x \sim y$ implies $y \sim x$.
To check transitivity, note that 
if $mx=m'x'$ and $m''x'=m'''x''$, then
	\[ m''mx = m''m'x' = m'm''x' = m'm'''x'' \]
where the middle equality uses commutativity of $M$. Hence $x \sim x''$.
%Next observe that for any $x,x' \in X$ and any $m \in M$ we have $x\sim x'$ if and  only if $x \sim mx'$.

%Note that $x \sim mx$ for all $x \in X$ and all $m \in M$. 
%So 
If $X_1$ is any equivalence class of $\sim$ then $MX_1 \subset X_1$.
Hence if $\mu$ is irreducible then $\sim$ has just one equivalence class. Conversely, if $\sim$ has just one equivalence class then for any partition of $X$ into nonempty subsets $X_1$ and $X_2$, we have $m_1 x_1 = m_2 x_2$ for some $x_1 \in X_1$, $x_2 \in X_2$ and $m_1,m_2 \in M$. If $MX_1 \subset X_1$ then $m_1 x_1 \in X_1$ and hence $MX_2 \not \subset X_2$.
\end{proof}

%In general, if $\mu$ is not irreducible then $X$ can be written as a disjoint union of irreducible components $X_\alpha$, which are the equivalence classes of  $\sim$.

Assume now that $M$ is finite. An \emph{idempotent} is an element $f \in M$ such that $ff=f$. Among the powers of any element $m\in M$ is an idempotent: by the pigeonhole principle, $m^j = m^k$ for some $j <k$, and then $m^\ell = m^{2\ell}$ where $\ell$ is any integer multiple of $k-j$ such that $\ell \geq j$.  Now consider the product of all idempotents in $M$:
	\[ e := \prod_{f \mbox{ \scriptsize idempotent} } f. \]
Note that $e$ is again an idempotent since $M$ is commutative. This \emph{minimal idempotent} $e$ is accessible from all of $M$: that is, $e \in mM$ for all $m \in M$.  

\begin{lemma}
%Let $M$ be a finite commutative monoid. Then
$eM$ is an abelian group with identity element $e$.
\end{lemma}

\begin{proof} %in fact e could be any idempotent
Since $e$ is an idempotent we have $e(em) = (ee)m=em$ for all $m \in M$, which verifies the identity axiom.
Since $e$ is accesible from all of $M$, we have $e \in (eme)M = (em)(eM)$ which verifies existence of inverses.
\old{
Consider the action of $M$ on itself by multiplication. This action is irreducible because $e$ is accessible from all of $M$.  For $m \in eM$ we have $m=em$ by part (4) of the above lemma.  Now by part (2), $e \in mM = (me)M = m(eM)$, which verifies the existence of inverses.
}
\end{proof}

In particular, $e \neq \eps$ unless $M$ itself is a group.

\old{ %%% new version uses idempotents instead of ideals
An \emph{ideal} of $M$ is a subset $I \subset M$ such that $M I \subset I$.  The intersection of ideals is again an ideal.  The minimal ideal
	\[ J = \bigcap_{\mbox{\scriptsize ideals } I \subset M} I \]
is a nonempty abelian group.  Denote its identity element by $e$.  Then $J=eM$.  In particular, $e\neq 0$ unless $J=M$.  One way to construct $e$ explicitly is 
	\[ e = \prod_{x \mbox{ \scriptsize idempotent} } x \]
where the product is over all idempotent elements of $x \in M$ (elements such that $xx=x$).  Since $M$ is commutative, the product of idempotents is again an idempotent.  A characterizing property of $e$ is that it is the unique idempotent accessible from all of $M$: that is, $ee=e$ and $e \in mM$ for all $m \in M$.
}%%%
	
\begin{lemma}
\moniker{Recurrent Elements Of A Monoid Action}
\label{l.recurrent}
Let $M$ be a finite commutative monoid and $\mu : M \times X \to X$ an irreducible action.  The following are equivalent for $x \in X$:
\begin{enumerate}
\item $x \in My$ for all $y \in X$
\item $x \in M(mx)$ for all $m \in M$
\item $x \in mX$ for all $m \in M$
\item $x \in eX$
\item $x = ex$
\end{enumerate}
\end{lemma}

\begin{proof}
(1) $\Rightarrow$ (2): trivial. \\
(2) $\Rightarrow$ (3): Since $M$ is commutative, $M(mx) = (Mm)x = (mM)x \subset Mx$. \\
(3) $\Rightarrow$ (4): trivial. \\
(4) $\Rightarrow$ (5): If $x=ey$, then $ex = e(ey) = (ee)y = ey = x$. \\
(5) $\Rightarrow$ (1): Here we use irreducibility. By Lemma~\ref{l.chainshortening}, given $x,y \in X$ there exist $m,m' \in M$ such that $mx = m'y$.  Let $m''$ be such that $m''m=e$.  If $x=ex$, then $x = m''mx = m''m'y \in My$.
\end{proof}

Any action of a finite commutative monoid
	\[ \mu: M \times X \to X \]
induces by restriction a corresponding group action
	\[ eM \times eX \to eX. \]
To see this, note that for any $m \in M$ and $x \in X$ we have $m(ex) = (me)x = (em)x = e(mx) \in eX$, so the action of $M$ on $X$ restricts to a monoid action of $M$ on $eX$.  Since $e(ex) = (ee)x = ex$, the element $e$ acts by identity on $eX$.  Since $eM$ is a group with identity element $e$, it follows that $eM \times eX \to eX$ is a group action. 

In fact slightly more is true. We say that $m \in M$ \emph{acts invertibly} on a subset $Y \subset X$ if the map $y \mapsto my$ is a bijection $Y \to Y$.

\begin{lemma}
\label{l.actsinvertibly}
Let $M$ be a finite commutative monoid and $\mu : M \times X \to X$ a monoid action.  Then every $m \in M$ acts invertibly on $eX$.
%note: does not require irreducible
\end{lemma}

\begin{proof}
For any $m \in M$ and $x \in X$ we have $(em)(ex) = (eme)x = (mee)x = (me)x = m(ex)$, so $em$ and $m$ have the same action on $eX$.  Since $eM \times eX \to eX$ is a group action, $em$ and hence $m$ acts invertibly on $eX$.
\end{proof}

We say that a monoid action $\mu : M\times X \to X$ is \emph{faithful} if there do not exist distinct elements $m,m' \in M$ such that $mx=m'x$ for all $x \in X$.  
%Equivalently, $\mu$ defines an injection of $M$ into the monoid $\End X$ of all self-maps of $X$.
The next theorem shows that relatively weak properties of a monoid action (faithful and irreducible) imply stronger properties of the corresponding group action (free and transitive).

Let $G$ be a group with identity element $e$.  Recall that a group action $G \times Y \to Y$ is called \emph{transitive} if $Gy = Y$ for all $y \in Y$, and is called \emph{free} if for all $g \neq e$ there does not exist $y \in Y$ such that $gy=y$.  If the action is both transitive and free, then for any two elements $y,y' \in Y$ there is a unique $g \in G$ such that $gy=y'$; in particular, $\# G = \# Y$.  

\begin{theorem}
\label{t.monoidtogroup}
\moniker{Group Actions Arising From Monoid Actions}
Let $M$ be a finite commutative monoid and $\mu : M \times X \to X$ an irreducible monoid action. Then the restriction of $\mu$ to $eM \times eX$ is a transitive group action
	 \[ e\mu: eM \times eX \to eX. \]
If $\mu$ is also faithful, then $e\mu$ is free.
\end{theorem}

\begin{proof}
To show transitivity, let 
	$ R = \bigcap_{y \in X} (My). $
Then for any $x \in eX$ we have 
	\begin{equation*} \label{e.orbitsqueeze} R \subset Mx = M(ex) = (Me)x = (eM)x = e(Mx) \subset eX. \end{equation*}  
But $R=eX$ by the equivalence of (1) and (4) in Lemma~\ref{l.recurrent}, so the above inclusions are equalities.  In particular, $(eM)x = eX$ for all $x \in eX$, which shows that $eM$ acts transitively on $eX$.

To show freeness, suppose that $g(ex) = ex$ for some $g \in eM$ and some $x \in X$.  
Now fix $y\in X$. By transitivity, $ey=h(ex)$ for some $h \in eM$, hence
	\[ gy = gey = ghex = hgex = hex = ey. \]
If $\mu$ is faithful, it follows that $g=e$.  Hence $eM$ acts freely on $eX$.
\end{proof}

\subsection{Dickson's Lemma}

The following lemma can be proved by induction on $k$ using the infinite pigeonhole principle.
We remark that it is also a case of the Hilbert basis theorem applied to the monomial ideal $(\mathbf{t}^{\xx_1},\mathbf{t}^{\xx_2},\ldots)$ in the polynomial ring $\Q[\mathbf{t}] = \Q[t_1,\ldots,t_k]$ (In fact the proof of the basis theorem now found in many textbooks, due to Gordan in 1900, uses this special case as a stepping stone; see \cite[Ex.\ 15.15]{Eis95} and \cite{DdJ98} for some history.)

\begin{lemma}
\moniker{Dickson's Lemma, \cite{Dic13}}
\label{l.dickson}
Let $k \geq 1$ be an integer.  For any sequence $\xx_1,\xx_2,\ldots \in \N^k$, there exist indices $m<n$ such that $\xx_m\leq \xx_n$ in the coordinatewise partial ordering.
\end{lemma}

\begin{comment} %%% bare hands proof
\begin{proof}
Induct on $k$. The base case $k=1$ follows from the fact that $\N$ is well-ordered.  If $k\geq 2$, then $\N^k$ is the union of $\xx_1+\N^k$ with a finite number of hyperplanes
    \[ H_{i,a}=\{\yy \in \N^k |\yy_i = a\} \] 
for $i=1,\ldots,k$ and $a=0,\ldots,\xx_{1i}-1$.  If $\xx_j \in \xx_1 + \N^k$ for some $j>1$, then $\xx_1 \leq \xx_j$.  Otherwise, by the infinite pigeonhole principle, one of the hyperplanes $H_{i,a}$ contains infinitely many terms of the sequence $\xx_j$. Since $H_{i,a} \isom \N^{k-1}$ as partially ordered sets, the proof is complete by the inductive hypothesis.
\end{proof}
\end{comment} %%%

See \cite{FFSS11} for an effective version, giving bounds on $n$ in terms of $k$ and the largest jump in the sequence $\xx_i$. Another bound is obtained in \cite{BS15} from a formalization of an existence proof.

\subsection{Toppling matrices}
\label{s.nonneg}

We will use the following form of the Perron-Frobenius theorem (see \cite[\textsection 8]{HJ90} for part (i) and \cite{Ash87} for part (ii)).  

\begin{lemma} \label{l.perronfrobenius}
\moniker{Perron-Frobenius}
Let $P$ be a square matrix with nonnegative real entries.  
\begin{itemize}
\item[(i)] $P$ has a nonnegative real eigenvector $\xx$ with nonnegative real eigenvalue $\lambda$, such that the absolute values of all other eigenvalues of $P$ are $\leq \lambda$.  
% If $P$ is irreducible, then $\xx$ has strictly positive entries.
\item[(ii)] If $P$ has rational entries and $\lambda$ is rational, then $\xx$ can be taken to have integer entries.  
\end{itemize}
\end{lemma}
% for part (ii), the conjugates of \xx are also eigenvectors, sum of the conjugates is rational, scale it to get an integer eigenvector (need to prove it is nonnegative). 

Following \cite{PS04}, we call $L$ a \emph{toppling matrix} if it satisfies the equivalent conditions of the following lemma.  See \cite[Theorem~4.3]{FK62} for a proof of the equivalence. Such matrices are also known as ``$M$-matrices.''  

For a vector $\xx$, we write $\xx > \zero$ to mean that all coordinates of $\xx$ are nonnegative and $\xx \neq \zero$.

\begin{lemma}\label{l.posdef} 
\moniker{Toppling Matrices}
Let $L$ be a square matrix with real entries such that 
%$L_{ii}>0$ for all $i$, and 
%%% this condition is not needed (it is implied by item 1).
$L_{ij}\le 0$ for all $i\ne j$. The following are equivalent.
\begin{enumerate}
\item All principal minors of $L$ are positive.

\item All eigenvalues of $L$ have positive real part.

\item There exists a vector $\xx>\zero$ such that $L\xx>\zero$.

\item There exists a vector $\yy>\zero$ such that $L^T\yy>\zero$.

\item $L$ is invertible, and all entries of $L^{-1}$ are nonnegative.
\end{enumerate}
\end{lemma}

\section{Abelian networks; Halting Dichotomy}
\label{s.networks}

We now recall the definition of an abelian network, refering the reader to \cite{part1} for details.
In an abelian network on a directed graph $G=(V,E)$, each vertex $v \in V$ has a \emph{processor} $\Proc_v$ which is an automaton with input alphabet $A_v$ and state space $Q_v$. For each letter $a \in A_v$ there is a state transition map $t_a : Q_v \to Q_v$, and these maps are required to commute: $t_a t_b = t_b t_a$ for all $a,b \in A_v$.  For each edge $(v,u) \in E$ there is a message-passing function $A_v \times Q_v \to A_u^*$ specifying the word (possibly empty) sent to processor $\Proc_u$ in the event that $\Proc_v$ in state $q \in Q_v$ processes letter $a \in A_v$. Each message-passing function is required to satisfy a commutativity condition: namely, if two input words to $\Proc_v$ are permutations of one another, then for each outgoing edge $(v,u)$ the resulting messages passed to $\Proc_u$ must be permutations of one another.
% if G has multiple edges, we could require this just for each outneighbor u instead of each edge (v,u).

To give a concrete example, let $L$ be an integer $V \times V$ matrix with positive diagonal entries nonpositive off-diagonal entries. The (locally recurrent) \emph{toppling network} $\Topp(L)$ has $A_v = \{v\}$ and $Q_v = \Z / L_{vv} \Z$. The state transition is $t_v(q) = q+1$ (mod $L_{vv}$). Whenever processor $\Proc_v$ transitions from state $L_{vv}-1$ to state $0$ it  sends $-L_{uv}$ letters $u$ to each processor $\Proc_u$ for $u \in V$. 
% and no letters otherwise.
The \emph{sandpile network} $\Sand(G)$ of a directed graph $G$ without self-loops is the special case where $L_{vv}$ is the outdegree of vertex $v$ and $-L_{uv}$ is the number of edges from $v$ to $u$.  
% if G has self-loops then $\Sand(G)$ is still a toppling network but it is not locally recurrent.
The commutativity conditions in the previous paragraph hold vacuously for a toppling network because it is \emph{unary}: each alphabet $A_v$ has just one letter. 
See \cite{part1} for many other examples of abelian networks, including non-unary examples.

The total state of an abelian network $\Net = (\Proc_v)_{v \in V}$ is described by an element $\qq \in Q := \prod_{v \in V} Q_v$ giving the internal states of the processors, together with a vector $\xx \in \Z^A$ where $A = \sqcup_{v \in V} A_v$ indicating how many letters of each type are waiting to be processed.  We use the notation $\xx.\qq$ for this pair. Note that $A$ is a disjoint union, so $\xx$ also specifies the locations of the letters (the $\xx_a$ letters $a$ are located at the unique vertex $v$ such that $a \in A_v$).

An \emph{execution} is a word $w = a_1\cdots a_r \in A^*$. It prescribes an order in which letters are to be processed. We write $\pi_w(\xx.\qq)$ for the result of executing $w$ starting from $\xx.\qq$. 
By definition $\pi_w$ is the composition $\pi_{a_1} \circ \cdots \circ \pi_{a_r}$. Each $\pi_a$ has three effects: change the internal state $\qq_v$ to $t_a \qq_v$,  where $v$ is the unique vertex such that $a \in A_v$; decrement $\xx_a$ by one; and increment each coordinate $\xx_b$ by the number of letters $b$ passed (as specified by the message passing function $T_{(v,u)}$ where $b \in A_u$). Note that decrementing $\xx_a$ may cause it to become negative, which is the reason for taking $\xx$ in $\Z^A$ rather than $\N^A$.
Writing $\pi_{a_1\cdots a_i}(\xx.\qq) = \xx^i.\qq^i$, we say that $w$ is \emph{legal} for $\xx.\qq$ if $\xx^{i-1}_{a_i} \geq 1$ for all $i=1,\ldots,r$. We say that $w$ is \emph{complete} for $\xx.\qq$ if $\xx^r \leq \zero$. The least action principle \cite[Lemma 4.3]{part1} says that if $w$ is any legal execution for $\xx.\qq$ and $w'$ is any complete execution for $\xx.\qq$, then $|w|_a \leq |w'|_a$ for all $a\in A$, where $|w|_a$ is the number of letters $a$ in the word $w$.
 A consequence of particular importance in this paper is the following.
 
\begin{lemma} \label{l.dichotomy}
\moniker{Halting Dichotomy, \cite[Lemma 4.4]{part1}}
For a given initial state $\qq$ and input $\xx$ to an abelian network $\Net$, either
\begin{enumerate}
\item There does not exist a finite complete execution for $\xx.\qq$; or
\item Every legal execution for $\xx.\qq$ is finite, and any two complete legal executions $w,w'$ for $\xx.\qq$ satisfy $|w|=|w'|$. 
\end{enumerate}
\end{lemma}

It follows readily from the axioms of an abelian network 
(see \cite[Lemma~4.2]{part1}) 
%\ref{l.piscommute}
that $\pi_a \circ \pi_b = \pi_b \circ \pi_a$ for all $a,b \in A$ and hence $\pi_w$ depends only on $|w|$. For $\yy \in \N^A$, write $\pi_\yy$ to mean $\pi_w$ for any word $w$ such that $|w|=\yy$.  
We record here two basic properties of $\pi_\yy$. 

\begin{lemma} \label{l.piprops}
For all $\yy,\zz \in \N^A$ we have
\begin{enumerate}
\item[(i)] $\pi_{\yy+\zz} = \pi_\yy \circ \pi_\zz$.
\item[(ii)] For all $\xx,\ww \in \Z^A$ and all $\qq \in Q$, if $\pi_\yy(\xx.\qq) = \xx'.\qq'$, then  \[ \pi_\yy((\xx+\ww).\qq) = (\xx'+\ww).\qq'. \]
\end{enumerate}
\end{lemma}

\begin{proof}
For part (i), if $|w|=\yy$ and $|w'|=\zz$ then $\pi_{\yy+\zz} = \pi_{ww'} = \pi_w \circ \pi_{w'} = \pi_\yy \circ \pi_\zz$.
 Part (ii) is immediate from 
\cite[eq.\ (4)]{part1}; 
%\eqref{e.piformula}.
it says that any additional letters $\ww$ present during an execution affect neither the messages passed $\xx'-\xx$ nor the final state $\qq'$.
\end{proof}

\section{Total kernel and production matrix}
\label{s.KP}

In this section we continue the development of the foundations of abelian networks begun in \cite{part1}. We associate two algebraic objects to an abelian network, the \emph{total kernel} $K$ and \emph{production matrix} $P$.  Only $P$ figures in our criterion for halting on all inputs, but we will see that $K$ is the natural domain of $P$ considered as a $\Z$-linear map. In the sequel \cite{part3}, both $K$ and $P$ play an essential role in analyzing the critical group of an abelian network that halts on all inputs.

\subsection{The local action}
\label{s.localaction}

Given $\xx \in \N^A$ and $\qq \in Q$, define
	\[ \xx \acts \qq := \pi_{\xx}(\xx.\qq). \]
In words, $\xx \acts \qq \in \N^A \times Q$ is the result of performing the following operation starting from state $\qq$: ``for each $a \in A$ add $\xx_a$ letters of type $a$ and process each letter once.''  If $\xx \acts \qq = \yy.\rr$ then message passing produced a total of $\yy_a$ letters of type $a$ for each $a \in A$, and the resulting state of the processor at vertex $v$ was $\rr_v$ for each $v \in V$.

For any $\zz \in \Z^A$ we also write
	\[ \xx \acts (\zz.\qq) := \pi_{\xx}((\xx+\zz).\qq). \]
The next lemma shows that $\acts$ defines a monoid action of $\N^A$ on $\Z^A \times Q$. We call this the \emph{local action} because each processor $\Proc_v$ processes only the letters $\xx_v$ that were added at $v$. 
%(It does not process any additional letters it may receive due to message passing from its neighbors, nor the letters $\zz_v$ that were ``already present.'')
%In \cite{part3} we will define a global action by the rule ``process all letters until halting.''
		
\begin{lemma}
\label{l.localaction}
For any $\xx,\yy \in \N^A$ and $\zz \in \Z^A$ and any $\qq \in Q$,
 	\[ \xx \acts (\yy \acts (\zz.\qq)) = (\xx + \yy) \acts (\zz.\qq). \]
\end{lemma}

\begin{proof}
Write $\yy \acts \qq = \yy'.\qq'$. By Lemma~\ref{l.piprops}(i), 
  since $\pi_{\xx+\yy} = \pi_\xx \circ \pi_\yy$, we have
\begin{align*}
	 (\xx + \yy) \acts (\zz.\qq)
	&= \pi_{\xx+\yy}((\xx+\yy+\zz).\qq) \nonumber \\
	&= \pi_\xx (\pi_\yy ((\xx+\yy+\zz).\qq)) \nonumber \\
	&=  \pi_\xx((\xx+\yy'+\zz).\qq') \nonumber \\
	&= \xx \acts ((\yy'+\zz).\qq') \label{e.partialaction} \\
%	&= \xx \acts \pi_\yy((\yy'+\zz).\qq) \\
	&= \xx \acts (\yy \acts (\zz.\qq)) \nonumber
\end{align*}
where in the third and last equalities we have used Lemma~\ref{l.piprops}(ii).
\end{proof}

The \emph{local monoid} $M_v$ of a vertex $v$ is the set of maps $Q_v \to Q_v$ generated by the maps $t_a$ for $a \in A_v$ under composition.
Write
	$ t_v : \N^{A_v} \to M_v $
for the monoid homomorphism sending $\one_a \mapsto t_a$ for $a \in A_v$.  Denote by \[ t: \N^A \to \prod_{v \in V} M_v \] the Cartesian product of the maps $t_v$.  Each $t_v$ is surjective by the definition of $M_v$, so $t$ is surjective.
% if A is infinite, then we interpret N^A here as a direct product.
 Note that if $\yy \acts \qq = \zz.\rr$, then $\rr = t(\yy)\qq$.    Given $\mm \in \prod_{v \in V} M_v$, write $\mm\qq = (\mm_v \qq_v)_{v \in V}$. In general, knowing $\yy \acts \qq = \zz.\rr$ does not determine $\yy \acts (\mm \qq)$, but the next lemma shows that it does in the case $\rr=\qq$.

\begin{lemma}
\label{l.manyamplifiers}
If $\yy \acts \qq = \zz.\qq$, then $\yy \acts (\mm\qq) = \zz.(\mm\qq)$ for all $\mm \in \prod_{v \in V} M_v$.
\end{lemma}
% only used in the case \mm = (e_v)_{v \in V} ?
% also only used for \Net finite?

\begin{proof}
Since $t$ is surjective there exists $\uu \in \N^A$ such that $t(\uu) = \mm$.  Then $\uu \acts \qq = \ww.\mm\qq$ for some $\ww \in \N^A$.  By Lemma~\ref{l.localaction},
	\[ \yy \acts (\ww.\mm\qq) = \yy \acts (\uu \acts \qq) = \uu \acts (\yy \acts \qq) = \uu \acts (\zz.\qq) = (\ww+\zz).\mm\qq. \]
From Lemma~\ref{l.piprops}(ii) it follows that $\yy \acts \mm\qq = \zz.\mm\qq$.
\end{proof}

\subsection{Production matrix}
\label{s.production}

From here on we assume that $\Net$ is a \emph{finite} abelian network:
that is, the underlying graph $G=(V,E)$ is finite, and the alphabet and state space of each vertex are finite sets.
%Some of the lemmas below hold more generally for locally finite $\Net$, provided we interpret $\N^A$ and $\Z^A$ as direct products (i.e., \emph{all} functions $A\to \N$ and $A \to \Z$ respectively).  
The main ingredients that rely on finiteness are the results from \textsection\ref{s.monoid}.

Let $e_v$ be the minimal idempotent of the local monoid $M_v$. The recurrent elements (Lemma~\ref{l.recurrent}) of the monoid action $M_v \times Q_v \to Q_v$ play a special role in this section.

\begin{definition}\label{d.recurrent}
A state $\qq \in Q$ is \emph{locally recurrent} if $\qq_v \in e_v Q_v$ for all $v \in V$.
(Equivalently, $\qq_v = e_v \qq_v$ for all $v \in V$.)
% These are equivalent even if the monoid action is not irreducible.
\end{definition}

By Lemma~\ref{l.actsinvertibly}, every $m \in M_v$ acts invertibly on $e_v Q_v$.  Thus for each $a \in A_v$ the map $q \mapsto T_v(a,q)$ is a permutation of $e_v Q_v$, so we have a group action
	\[ \Z^{A_v} \times e_v Q_v \to e_v Q_v. \]
Let $K_v$ be the set of vectors in $\Z^{A_v}$ that act as the identity on $e_v Q_v$.

\begin{definition}\label{d.totalkernel}
The \emph{total kernel} of $\Net$ is the subgroup of $\Z^A$ given by
	\[ K = \prod_{v \in V} K_v. \]
\end{definition}
	
\begin{lemma}
\label{l.fullrank}
%If $\Net$ is locally finite, then each $K_v$ is a subgroup of finite index in $\Z^{A_v}$.
If $\Net$ is finite, then 
$K$ is a subgroup of finite index in $\Z^A$. In particular, $K$ is generated as a group by $K \cap \N^A$.
% To make the last sentence true in the locally finite case, we should interpret \N^A as a direct product.
\end{lemma}

\begin{proof}
Since 
%$\Net$ is locally finite, 
each $e_v Q_v$ is a finite set, for any $\xx \in \Z^{A_v}$ we have $n\xx \in K_v$ for some $n \geq 1$. Thus $K_v$ has finite index in $\Z^{A_v}$.  
Since 
%$\Net$ is spatially finite, 
$V$ is a finite set, $K = \prod_{v\in V} K_v$ has finite index in $\Z^A$.
In particular, $K$ contains a vector with all coordinates strictly positive, which implies that $K$ is generated as a group by $K \cap \N^A$.
% Indeed, for any k in K we have k+np >=0 for some n where p is the strictly positive vector.  Both k+np and np are in K \cap N^A, and k = (k+np) - np.
\end{proof}

Note that
	\begin{equation} \label{e.localkernel} K \cap \N^A = \{\xx \in \N^A \mid t(\xx) \qq = \qq \mbox{ for all locally recurrent } \qq \in Q \}. \end{equation}
Fix a locally recurrent state $\qq \in Q$.  For any $\kk \in K \cap \N^A$ we have
	\begin{equation} \label{e.production} \kk \acts \qq =  P_{\qq}(\kk).\qq \end{equation}
for some vector $P_{\qq}(\kk) \in \N^A$.  Next we show that
	\[ P_{\qq} : K \cap \N^A \to \N^A \]
extends to a group homomorphism.

\begin{lemma}
\label{l.extends}
Let $\Net$ be a finite abelian network.  Then $P_{\qq}$ extends to a group homomorphism $K \to \Z^A$.  
\end{lemma}

\begin{proof}
Write $P=P_\qq$.
Let $\kk_1, \kk_2 \in K \cap \N^A$.  By Lemma~\ref{l.localaction},
	\[ (\kk_1+\kk_2) \acts \qq = \kk_1 \acts (P(\kk_2).\qq) = (P(\kk_1)+P(\kk_2)).\qq \]
hence
	\begin{equation} \label{e.additivity} P(\kk_1 + \kk_2) = P(\kk_1) + P(\kk_2). \end{equation}

By Lemma~\ref{l.fullrank}, every $\kk \in K$ can be written as $\kk_1 - \kk_2$ for $\kk_1, \kk_2 \in K \cap \N^A$.  Define $P(\kk) = P(\kk_1) - P(\kk_2)$.  Equation \eqref{e.additivity} now implies that this extension is well defined and a group homomorphism.
%If also $\xx = \kk'_1 - \kk'_2$, then by Lemma~\ref{l.additive},
%	\[ P(\kk_1) - P(\kk_2) - (P(\kk'_1) - P(\kk'_2)) = P(\kk_1 + \kk'_2) - P(\kk'_1 + \kk_2) = 0 \]
%which shows the extension is well defined.
\end{proof}

By tensoring $P_{\qq} : K \to \Z^A$ with $\Q$, we obtain a linear map $P_{\qq} : \Q^A \to \Q^A$. To be more explicit, for any $\xx \in \Q^A$, by Lemma~\ref{l.fullrank} there is an integer $n \geq 1$ such that $n \xx \in K$, and we define
	\[ P_{\qq}(\xx) := \frac1n P_{\qq}(n \xx). \]
So far we have defined $P_{\qq}$ only for locally recurrent $\qq$.  We extend the definition to all states $\qq =(q_v)_{v\in V}$ by setting $P_{\qq} := P_{\widehat{\qq}}$, where $\widehat{\qq} = (e_v q_v)_{v \in V}$. 
 
\begin{definition}
\label{d.production}
The \emph{production matrix} of a finite abelian network $\Net$ with initial state $\qq$ is the matrix of the linear map $P_{\qq} : \Q^A \to \Q^A$.
\end{definition}

The $(a,b)$ entry $p_{ab}$ of the production matrix says ``on average'' how many letters~$a$ are created by processing the letter $b$: specifically, if $n\basis_b \in K$, then $p_{ab} = \frac{1}{n} p_{ab}(n)$, where $p_{ab}(n) = \mathbf{N}(b^n,q)_a$ is the number of $a$'s created by executing the word~$b^n$.
 We have chosen the term ``production matrix'' to evoke \cite{DFR05}.  Indeed the succession rules studied in that paper can be modeled by an abelian network whose underlying graph is a single vertex with a loop.

We remark that the production matrix can be also defined for some networks that are not 
%locally
finite by setting
	\[ p_{ab} := \lim_{n \to \infty} \frac{1}{n} p_{ab}(n) \]
if this limit exists.

\subsection{Local components}\label{s.localcomp}

An abelian finite automaton $\Proc$ with state space $Q$ and alphabet $A$ is specified by transition maps $t_a : Q \to Q$ for $a \in A$, such that $t_a \circ t_b = t_b \circ t_a$ for all $a,b \in A$.  The \emph{transition monoid} of $\Proc$ is the submonoid $M \subset \End(Q)$ generated by $\{t_a\}_{a \in A}$, where $\End(Q)$ denotes the monoid of all set maps $Q \to Q$ under composition. (For example, the local monoid $M_v$ of \textsection\ref{s.localaction} is the transition monoid of $\Proc_v$.)
%(In \cite{part3}\textsection\ref{s.critical} we will study the ``global'' transition monoid of an abelian network that halts on all inputs.)

By construction, the transition monoid of $\Proc$ has a faithful action $M \times Q \to Q$.
We say that $\Proc$ is \emph{irreducible} if this monoid action is irreducible (\textsection \ref{s.monoid}).  The \emph{irreducible components} $\Proc_\alpha$ of $\Proc$ are the automata with alphabet $A$ and state space $Q_\alpha$, where $Q = \sqcup Q_\alpha$ is the partition of $Q$ into equivalence classes under the relation $q \sim q'$ if there exist $m,m' \in M$ such that $mq = m'q'$ (see Lemma~\ref{l.chainshortening}).

We say that an abelian network $\Net = (\Proc_v)_{v \in V}$ is \emph{locally irreducible} if each processor $\Proc_v$ is irreducible.
%(i.e., for all $v\in V$ the monoid action $M_v \times Q_v \to Q_v$ is irreducible).  
The \emph{local components} of $\Net$ are the abelian networks $\Net_{\alpha} = (\Proc_v^{\alpha_v})_{v \in V}$ where each $\Proc_v^{\alpha_v}$ is an irreducible component of $\Proc_v$.  Note that the local components have the same underlying graph $G$ and alphabet $A$, with a possibly smaller state space at each vertex.

\begin{comment}
\begin{example}
If $\Net$ is a height arrow network (\cite{DR02} \textsection \ref{s.heightarrow}), a vertex starting in state $(q,c)$ can only access states of the form $(q', c')$ where $q' = q+n\tau_v \pmod{d_v}$ for some $n\in \N$. If 
$g_v := \gcd(d_v,\tau_v)>1$, then $Q_v$ is not irreducible. Each processor $\Proc_v$ has $g_v$ irreducible components, and $\Net$ has $\prod_{v \in V} g_v$ local components.
\end{example}
\end{comment}

The next lemma should be compared with \eqref{e.localkernel}.

\begin{lemma} % this lemma could be moved to another section 
\label{l.free}
If $\Net$ is finite and locally irreducible, then for any fixed locally recurrent $\qq \in Q$ we have
	\[ K \cap \N^A = \Set{\kk \in \N^A}{t(\kk)\qq = \qq}. \]
\end{lemma}

\begin{proof}
%If $\kk \in K$, then $t(\kk)\qq =\qq$ for all locally recurrent $\qq$.  Conversely,
By Theorem~\ref{t.monoidtogroup} the action of $e_v M_v$ on $e_v Q_v$ is free.
If $\kk \in \N^A$ and $t(\kk)\qq = \qq$ for one locally recurrent $\qq$, then for all $v \in V$ 
	\[  t_v(\kk_v) e_v \qq_v = t_v(\kk_v) \qq_v = \qq_v. \]
By freeness it follows that $t_v(\kk_v) e_v = e_v$, so $t(\kk) \rr = \rr$ for all locally recurrent $\rr$.  Hence $\kk \in K$.
\end{proof}

For $\qq,\rr \in Q$ we write $\qq \sim \rr$ if $\qq$ and $\rr$ belong to the same local component (that is, $\qq_v \sim \rr_v$ for all $v \in V$).  We denote by $\Net_\qq$ the local component of $\Net$ containing state $\qq$. The next lemma shows that within a local component all states have the same production matrix.

\begin{lemma}
\label{l.indepofq}
%Let $\Net$ be a finite abelian network.
~
\begin{enumerate}
\item[(i)] In a locally irreducible finite abelian network, the production matrix $P_{\qq}$ does not depend on the initial state $\qq$.
\item[(ii)] In any finite abelian network, if $\qq \sim \rr$ then $P_\qq = P_\rr$.
\end{enumerate}
\end{lemma}

\begin{proof}
(i) Let $\qq_1,\qq_2 \in Q$ be locally recurrent states.  
%If $\Net$ is locally irreducible, then 
By Theorem~\ref{t.monoidtogroup} each group action $e_v M_v \times e_v Q_v \to e_v Q_v$ is transitive, so there exists $\xx \in \N^A$ such that $t(\xx) \qq_1 = \qq_2$.  Let $\zz$ be such that
	\[ \xx \acts \qq_1 = \zz.\qq_2. \]
Fix $\kk \in K \cap \N^A$, and let $\yy_i = P_{\qq_i}(\kk)$ for $i=1,2$.  Then
	\[ \xx \acts (\kk \acts \qq_1) = \xx \acts (\yy_1.\qq_1) = (\yy_1+\zz).\qq_2 \]
while
	\[ \kk \acts (\xx \acts \qq_1) = \kk \acts (\zz.\qq_2) = (\yy_2+\zz).\qq_2. \]
By Lemma~\ref{l.localaction},
	\[ \kk \acts (\xx \acts \qq) = \xx \acts (\kk \acts \qq) \]
hence $\yy_1 + \zz = \yy_2 + \zz$, and hence $\yy_1 = \yy_2$.  Since $\kk \in K \cap \N^A$ was arbitrary we conclude that $P_{\qq_1} = P_{\qq_2}$.
%This shows that $P_q(\kk)$ does not depend on $q$.

Part (ii) follows by applying part (i) to the local component $\Net_\qq$.
\end{proof}

\subsection{Strong components}
\label{s.strong}

Let $\Net$ be a locally irreducible finite abelian network with production matrix $P = (p_{ab})_{a,b \in A}$. In this section we describe another way to break $\Net$ into smaller pieces by reducing its alphabet $A$. (The results of this section are not used anywhere else in this paper, so the reader who wishes to skip to the halting problem in \textsection\ref{s.halting} can safely do so.)

\begin{figure}
\centering
\medskip
\includegraphics[height=.5\textheight]{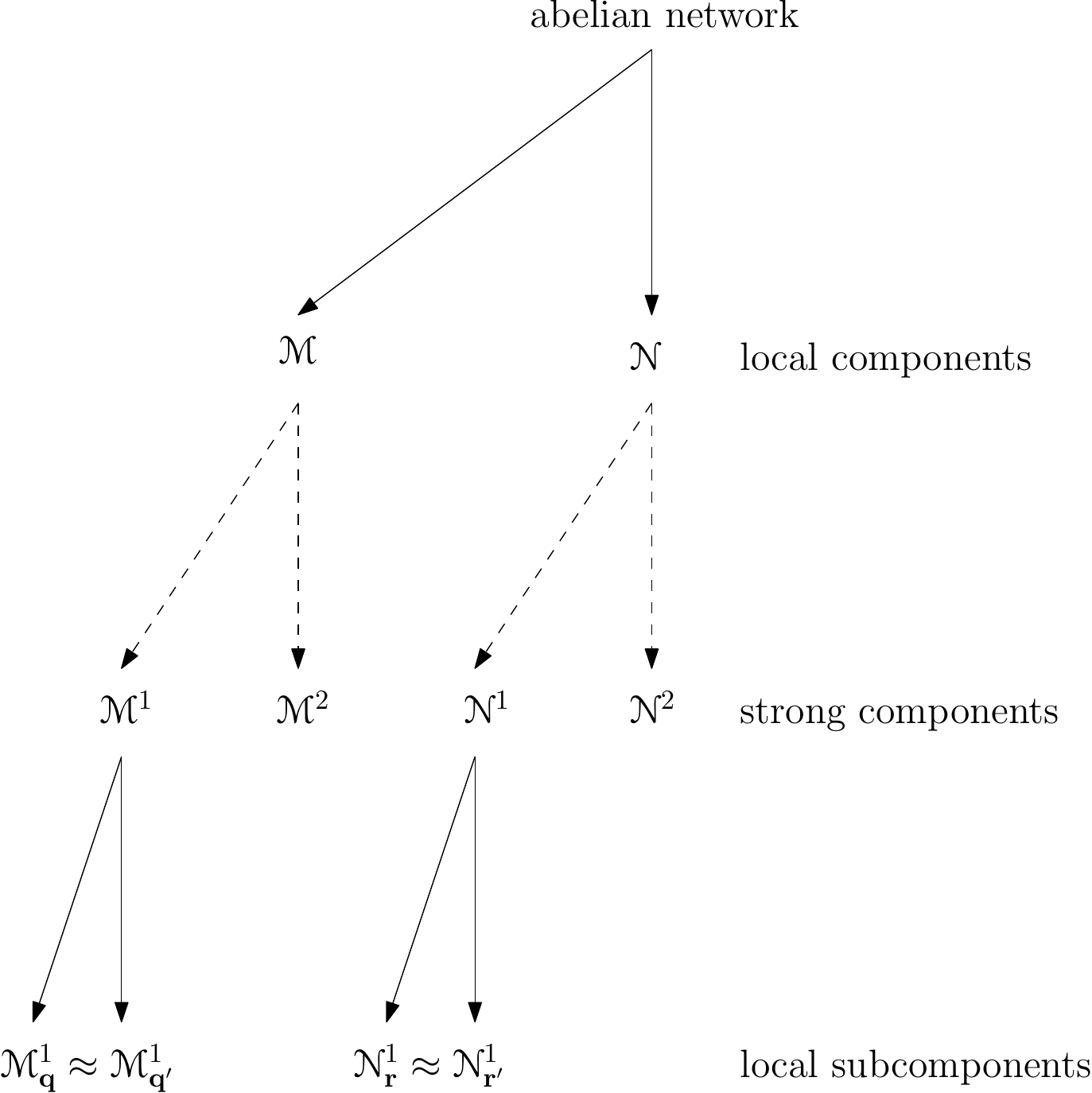}
\caption{The local components of an abelian network may decompose into strong components, which may further decompose further into local subcomponents. Solid arrows represent restrictions of the state space, and dashed arrows represent restrictions of the alphabet.}
\label{f.components}
\end{figure}

\begin{definition} \label{d.productiongraph}
The \emph{production graph} $\Gamma = \Gamma(\Net)$ of $\Net$ is the directed graph
with vertex set $A$ and edge set $\{(a,b) \mid p_{ba} >0 \}$. 
% note: not well defined if \Net is not locally irreducible!
\end{definition}

Write $a \dashrightarrow b$ if there is a directed path in $\Gamma$ from $a$ to $b$.  The \emph{strong components} of $\Gamma$ are the equivalence classes of the relation $\{(a,b) \mid a \dashrightarrow b \mbox{ and } b \dashrightarrow a\}$.

\begin{definition}
The \emph{strong components} of $\Net$ are the subnetworks $\Net^1, \ldots, \Net^s$ with alphabets $A^1, \ldots, A^s$, where $A^1, \ldots, A^s$ are the strong components of $\Gamma$.
\end{definition}
% components all have the same underlying graph G.  some alphabets A^i \cap A_v may be empty.

Note that the local components of \textsection\ref{s.localcomp} are defined by restricting the state space, whereas the strong components are defined by restricting the alphabet.  The strong components of a locally irreducible network need not be locally irreducible (Figure~\ref{f.components}).  However, Lemma~\ref{l.subcomponents} below shows that the local components of a strong component do not decompose any further. Moreover, all local components of a given strong component are homotopic.

\begin{definition}
\label{d.homotopy}
(Homotopy)
Locally irreducible abelian networks $\Net$ and $\Net'$ on the same graph with the same total alphabet are called \emph{homotopic}, written $\Net \approx \Net'$, if they have the same total kernel $K$ and the same production matrix $P$.  
\end{definition}

The reason for calling this homotopy comes from a method of diagraming states described in \cite{part1}. The locally recurrent states of each processor correspond to points on the discrete torus $\Z^{A_v} / K_v$, with state transitions given by adding a basis vector $\one_a$ for $a \in A_v$. Message passing can be visualized by surfaces cutting between the states, with each surface labeled by the letter to be passed. 
%when transitioning across that surface.  
If $\Net$
%= (\Proc_v)_{v \in V}$ 
and $\Net'$
%=(\Proc'_v)_{v \in V}$ 
have the same total kernel $K$, then for each vertex $v$ the state diagrams of processors $\Proc_v$ and $\Proc'_v$ live on the same discrete torus.  Then $\Net$ and $\Net'$ have the same production matrix if and only if for each $v$ the state diagram of $\Proc'_v$ can be obtained from the state diagram of $\Proc_v$ by altering message surfaces without changing the homotopy type of any surface.

\begin{lemma}
\label{l.subcomponents}
Let $\Net$ be a locally irreducible finite abelian network with strong components $\Net^i$, total kernel $K$ and production matrix $P$.  Let $\Net^i_\qq$ be the local component of $\Net^i$ containing state $\qq$.  The following hold for all $\qq,\qq' \in Q$.
\begin{enumerate}
\item[(i)] The total kernel of $\Net^i_\qq$ is the group generated by $K \cap \N^{A^i}$.
% is this equal to K \cap \Z^{A^i} ?
\item[(ii)] The production matrix of $\Net^i_\qq$ is the $A^i \times A^i$ submatrix $P_{ii}$ of $P$.  
\item[(iii)] $\Net^i_\qq$ has only one strong component. 
\item[(iv)] $\Net^i_\qq \approx \Net^i_{\qq'}$.
\end{enumerate}
\end{lemma}

\begin{proof}
%Let $e_v^i$ be the minimal idempotent of the local monoid $M_v^i$ of $\Net^i$. 
Since $\Net^i_\qq = \Net^i_{\widehat \qq}$ where $\widehat \qq = (e^i_v\qq_v)_{v\in V}$, we may assume that $\qq$ is a locally recurrent state of $\Net^i$.

(i) By Lemma~\ref{l.fullrank}, the total kernel $K_\qq^i$ of $\Net_\qq^i$ is generated as a group by $\K_\qq^i \cap \N^{A^i}$. Since both $\Net^i_\qq$ and $\Net$ are locally irreducible, by Lemma~\ref{l.free} we have for $\kk \in \N^{A^i}$ 
	\[ \kk \in K^i_\qq \;\iff\; t(\kk)\qq = \qq \;\iff\; \kk \in K. \]
%(For the first equivalence we have used that $t^i_\qq (\kk) = t(\kk)|_{Q_\qq}$.) 
Hence \[ \K_\qq^i \cap \N^{A^i} = K \cap \N^{A^i}. \]
%This shows that $K^i_\qq$ is the group generated by $K \cap \N^{A^i}$.  For any $\xx \in K \cap \Z^{A^i}$ we have $\xx = \kk - \kk'$ for $\kk, \kk' \in K \cap \N^A$, but can we choose \kk, \kk' \in \N^{A^i} ?

(ii) Write $\acts$ for the local action of $\Net$ and $\acts^i$ for the local action of $\Net^i$.  
For $\kk \in \N^{A^i}$ the only difference between these actions is that $\kk \acts \qq$ may produce some letters in $A-A^i$ in addition to the letters in $A^i$ produced by $\kk \acts^i \qq$.
 Hence, writing $P_\qq^i$ for the production matrix of $\Net_\qq^i$, if $\kk \in K \cap \N^{A^i}$ then $\kk \acts \qq = \yy.\qq$, where $\yy = P_\qq^i \kk + \zz$ for some $\zz \in \N^{A-A^i}$.  
  % this uses part (i)
Letting $\rr = (e_v\qq_v)_{v \in V}$, we have $\kk \acts \rr = \yy.\rr$ by Lemma~\ref{l.manyamplifiers}. Since $\rr$ is locally recurrent for~$\Net$, it follows that $\yy = P\kk$.  Writing $\rho_i$ for the projection $\N^A \to \N^{A^i}$, we conclude that 
	\[ P_\qq^i \kk = \rho_i (P\kk) = P_{ii} \kk \]
and hence $P_\qq^i = P_{ii}$.

(iii) By part (ii), the production graph of $\Net_\qq^i$ is the strong component $A^i$ of $\Gamma$, so $\Net_\qq^i$ has only one strong component.

(iv) By parts (i) and (ii), $K^i_\qq$ and $P^i_\qq$ do not depend on $\qq$, so all local components of $\Net^i$ are homotopic.
\end{proof}

The strong components are partially ordered by the accessibility relation $\dashrightarrow$.  If we label them so that $A^i \dashrightarrow A^j$ implies $i \geq j$, then the production matrix is block triangular
	\[ P = \left( \begin{array}{ccccc} 
	P_{11} & P_{12} & \cdots & P_{1s} \\
	0 & P_{22} & \cdots & P_{2s} \\
	\vdots & \vdots & \ddots & \vdots \\
	0 & 0 & \cdots & P_{ss}
	\end{array} \right) \]
and the diagonal block $P_{ii}$ is the production matrix of $\Net^i$.  

%The Perron-Frobenius eigenvector $\xx^i$ of $P_{ii}$ has strictly positive entries.  

% Note: K contains \prod K^i but they are not always equal.  To see this, take a,b \in A_v and a,b in different strong components.  If K_v is non-rectangular, then K_v is strictly larger than K_v^a \times K_v^b.

\section{Halting criterion}
\label{s.halting}

Let $\Net$ be an abelian network with total state space $Q$, and fix a state $\qq \in Q$.
If case (2) of Lemma~\ref{l.dichotomy} holds for all inputs $\xx \in \N^A$, then we say that $\Net$ \emph{halts on all inputs to initial state $\qq$}.  If this is the case for all $\qq \in Q$, then we say that $\Net$ \emph{halts on all inputs}.
The main result of this section is Theorem~\ref{t.halting}, which gives an  efficient way to decide whether a finite abelian network halts on all inputs to a given initial state.

If $\xx \leq \yy$, then any legal execution for $\xx.\qq$ is a legal execution for $\yy.\qq$.  It follows that halting is a monotone property: 
\begin{equation}
\label{e.haltingismonotone}
\mbox{If $\Net$ halts on input $\yy.\qq$, then $\Net$ halts on all inputs $\xx.\qq$ for $\xx \leq \yy$.}
\end{equation}

Recall the equivalence relation $\sim$ on $Q$ introduced in \textsection\ref{s.localcomp}.

\begin{lemma}
\label{l.haltingequiv}
If $\qq_1 \sim \qq_2$, then $\Net$ halts on all inputs to initial state $\qq_1$ if and only if $\Net$ halts on all inputs to initial state $\qq_2$.
\end{lemma}

\begin{proof}
If $\qq_1 \sim \qq_2$ then $t(\yy_1)\qq_1 = t(\yy_2)\qq_2$ for some $\yy_1,\yy_2 \in \N^A$. Thus it suffices to show for all $\yy \in \N^A$ and all $\qq \in Q$ that 
$\Net$ halts on all inputs to initial state $\qq$ if and only if $\Net$ halts on all inputs to initial state $t(\yy)\qq$.

Let $\yy \acts \qq = \zz.\rr$. Then $\rr=t(\yy)\qq$.  For any $\xx \in \N^A$, we have
  \[ (\xx + \yy) \acts \qq = \xx \acts (\yy \acts \qq) = \xx \acts (\zz.\rr) \]
  so both $(\xx+\yy).\qq$ and $(\xx+\zz).\rr$ have legal executions resulting in $(\xx+\yy) \acts \qq$.
 Thus $\Net$ halts on input $(\xx+\yy).\qq$ if and only if $\Net$ halts on input $(\xx + \zz).\rr$.  
By monotonicity \eqref{e.haltingismonotone}, it follows that $\Net$ halts on all inputs to initial state $\qq$ if and only if $\Net$ halts on all inputs to initial state $\rr$. 
 \end{proof}

\begin{definition} 
\label{d.amp}
A state $\xx.\qq$ is an \emph{amplifier} if $\xx \in \N^A$ and there exists a nonempty legal execution $w$ from $\xx.\qq$ such that $\pi_w (\xx.\qq) = \yy.\qq$ for some $\yy \geq \xx$.
\end{definition}
% note this rules out x=0 since only the empty execution is legal from 0.q

\begin{definition} 
\label{d.strongamp}
A state $\xx.\qq$ is a \emph{strong amplifier} if $\xx \in \N^A - \{\mathbf{0}\}$ and $\xx \acts \qq = \yy.\qq$ for some $\yy \geq \xx$.
\end{definition}

In words, a strong amplifier is a pair $\xx.\qq$ with the property that after processing all letters once, the network has returned to the same state $\qq$ with at least as many letters of each type as before.

\begin{example}
In the sandpile network $\Sand(G)$ of an undirected graph with no sink, let $\xx = (d_v)_{v \in V}$ be the configuration where each vertex has the same number of letters (``chips'') as its degree.  For any initial state $\qq$, processing all letters once causes each vertex to topple once, so that each vertex $v$ receives one letter from each of its $d_v$ neighbors.  Hence $\xx \acts \qq = \xx.\qq$, and $\xx.\qq$ is a strong amplifier.
\end{example}

\begin{lemma}
\label{l.amp}
The following are equivalent for a finite abelian network $\Net$.
\begin{enumerate}
\item $\Net$ has an amplifier.
\item $\Net$ has a strong amplifier.
\item $\Net$ fails to halt on some input.
\end{enumerate}
\end{lemma}

\begin{proof}
If $\Net$ has an amplifier $\xx.\qq$, then there is a legal execution $\pi_\uu(\xx.\qq) = \yy.\qq$ for some $\yy \geq \xx$ and $\uu \in \N^A -\{\zero\}$.  Then $\uu \acts \qq = \pi_\uu(\uu.\qq) = (\uu + \yy-\xx).\qq$, where the second equality follows from Lemma~\ref{l.piprops}(ii). Therefore $\uu.\qq$ is a strong amplifier, which shows that (1) $\Rightarrow$ (2).

Next if $\Net$ has a strong amplifier $\uu.\qq$, then there is a legal execution $w$ 
with $|w|=\uu \in \N^A - \{\zero\}$
%but this doesn't matter: we don't actually care that the amplifier is strong
starting with $\uu.\qq$ and ending with $(\uu+\vv).\qq$ for some $\vv \geq \zero$.  
Then for any $n \geq 0$ the same $w$ is a legal execution starting with $(\uu + n\vv).\qq$ and ending with $(\uu + (n+1)\vv).\qq$.  Hence $w^n$ is a legal execution starting from $\uu.\qq$ for all $n \geq 0$.  Since there exist arbitrarily long legal executions, we conclude from Lemma~\ref{l.dichotomy} that $\Net$ does not halt on input $\uu.\qq$, which shows that (2) $\Rightarrow$ (3).

Lastly, suppose that $\Net$ fails to halt on some input $\yy.\qq$. 
  Then there is an infinite word $w_1 w_2 \cdots$ such that $w_1 \cdots w_n$ is a legal execution for all $n \geq 1$. Let $\yy_n.\qq_n = \pi_{w_1 \cdots w_n} (\yy.\qq)$.   Since the total state space $Q$ is finite, there exists $\qq \in Q$ such that $\qq_n=\qq$ for infinitely many~$n$.  By Dickson's Lemma~\ref{l.dickson}, there exist indices $j < k$ such that $\qq_j = \qq_k = \qq$ and $\yy_j \leq \yy_k$.  This $\yy_j.\qq_j$ is an amplifier, which shows that (3) $\Rightarrow$ (1).
%also, $\yy_j + \cdots + \yy_{k-1}$ is a strong amplifier.
\end{proof}

The next lemma is a variant of Lemma~\ref{l.amp} with distinguished initial state. Recall from \textsection\ref{s.localcomp} the local component $\Net_\qq$ containing state $\qq$.

\begin{lemma}
\label{l.ampwithq}
The following are equivalent for a finite abelian network $\Net$ and state~$\qq$.
\begin{enumerate}
\item $\Net_\qq$ has an amplifier.
\item $\Net_\qq$ has a strong amplifier.
\item $\Net$ fails to halt on some input to initial state $\qq$.
\end{enumerate}
\end{lemma}

\begin{proof} 
By Lemma~\ref{l.amp} it suffices to show that (3) is equivalent to the statement that $\Net_\qq$ fails to halt on some input.
Any execution for $\xx.\qq$ in $\Net$ is also an execution for $\xx.\qq$ in $\Net_\qq$, so (3) is equivalent to the statement that $\Net_\qq$ fails to halt on some input to initial state $\qq$.  By Lemma~\ref{l.haltingequiv}, if $\Net_\qq$ fails to halt on some input, then $\Net_\qq$ fails to halt on some input to initial state $\qq$, which completes the proof.
\end{proof}

%\subsection{Criterion for halting}

\begin{theorem}
\label{t.halting}
\moniker{Halting Criterion 1}
A finite abelian network $\Net$ halts on all inputs to initial state $\qq$ if and only if every eigenvalue of the production matrix $P_\qq$ has absolute value strictly less than~$1$.
\end{theorem}

\begin{proof}
%The local component $\Net_\qq$ includes all states accessible from $\qq$, so we may assume without loss of generality that $\Net$ is locally irreducible. 
Let $P=P_\qq$, and let $\xx$ and $\lambda$ be the Perron-Frobenius eigenvector and eigenvalue of $P$; if $\lambda=1$ then we may take $\xx$ to have integer entries (Lemma~\ref{l.perronfrobenius}).  By Lemma~\ref{l.ampwithq} it suffices to show that (1) if $\lambda \geq 1$, then $\Net_\qq$ has an amplifier; and (2) if $\Net_\qq$ has a strong amplifier, then $\lambda \geq 1$.

(1)  If $\lambda \geq 1$, then there is a vector $\yy \in \Q^A$ such that $\xx \leq \yy \leq \lambda \xx$.  Then $P \yy \geq P \xx = \lambda \xx \geq \yy$. Choosing $n \geq 1$ such that $n \yy \in K$, we have $n\yy \acts \qq = P(n \yy).\qq$, so $n \yy.\qq$ is an amplifier.
% in fact, strong

(2) If $\Net_\qq$ has a strong amplifier $\yy.\qq'$, then by Lemma~\ref{l.manyamplifiers}, $\yy.\mm \qq'$ is also a strong amplifier for any $\mm \in \prod_{v \in V} M_v$.  In particular, taking $\mm=(e_v)_{v \in V}$ we obtain a strong amplifier $\yy.\rr$ with $\rr$ locally recurrent and $\rr \sim \qq' \sim \qq$.
By Lemma~\ref{l.indepofq}, $P_\qq (\yy) = P_{\rr}(\yy) \geq \yy$, which shows that $\lambda \geq 1$.
\end{proof}

\begin{remark}
In the case that $\Net$ is locally irreducible, the production matrix $P_\qq$ does not depend on $\qq$ by Lemma~\ref{l.indepofq}.  In this case Theorem~\ref{t.halting} gives a criterion for $\Net$ to halt on all inputs  regardless of initial state.
\end{remark}
 
If $P_\qq$ has Perron-Frobenius eigenvalue $\lambda \geq 1$, then $\Net$ runs forever on some inputs to initial state $\qq$. %and halt on other inputs.  In the next section we examine how to tell which is the case.
In the next section we examine how to tell whether this is the case for a given input.

\section{Laplacian matrix; Sandpilization}
\label{s.laplacian}

An abelian network $\Net$ is called \emph{locally finite} if the state space and alphabet of each processor are finite.
Let $\Net$ be a locally finite and locally irreducible abelian network with total alphabet $A = \sqcup A_v$ and total kernel $K = \prod K_v$.
Then each $K_v$ has finite index in $\Z^{A_v}$.  For each letter $a \in A_v$, let $r_a$ be the smallest positive integer such that $r_a 1_a \in K_v$.  
%The vector $\rr = (r_a)_{a \in A}$ is called the \emph{reset vector}. 

Denote by $D$ the $A \times A$ diagonal matrix with diagonal entries $r_a$, and by $I$ the $A \times A$ idenity matrix. Let $P$ be the production matrix of $\Net$, which is well defined by Lemma~\ref{l.indepofq}.

\begin{definition}
\label{d.laplacian}
The \emph{Laplacian} of $\Net$ is the $A \times A$ matrix
	\[ L = (I-P)D. \]
\end{definition}

\begin{lemma}
$L$ has integer entries.
\end{lemma}

\begin{proof}
If $\xx \in \N^A$ then $D\xx \in K$, so $(D\xx) \acts \qq = (PD\xx).\qq$ by the definition of the production matrix $P$. In particular, $PD\xx \in \N^A$.  It follows that $L\xx \in \N^A$ for all $\xx \in \N^A$, so $L$ has integer entries.
\end{proof}

The class of toppling matrices, defined in \textsection\ref{s.nonneg}, is one of several extensions of the notion of ``positive definite'' to non-symmetric matrices.  The next lemma encapsulates how toppling matrices arise in our setting. 

%%In particular, the damped Laplacian $\Delta^s$ of a connected graph (or more generally, of a directed graph in which every vertex has a directed path to $s$) is a toppling matrix. It corresponds to the special case of a sandpile network with a sink.  Toppling matrices are a bit more general than graph Laplacians in that they need not have all column sums nonnegative. Columns with negative sum correspond to productive vertices.  

\begin{lemma}
\label{l.mfrompositive}
Let $L=(I-P)D$, where $P$ is a nonnegative $n\times n$ matrix, $D$ is an $n \times n$ positive diagonal matrix, and $I$ is the $n \times n$ identity matrix.
Then $L$ is a toppling matrix if and only if the Perron-Frobenius eigenvalue of $P$ is strictly less than $1$.
\end{lemma}

\begin{proof}
The off-diagonal entries of $L$ are nonpositive, so Lemma~\ref{l.posdef} applies.  Let $\xx$ and $\lambda$ be a Perron-Frobenius eigenvector and eigenvalue of $P$, and let $\yy = D^{-1} \xx$. Then $\yy > \zero$ and
	\begin{equation} \label{ellwhy} L\yy = (I-P)\xx = (1-\lambda)\xx. \end{equation}
If $\lambda< 1$, then $L\yy >\zero$, so $L$ is a toppling matrix. Conversely, if $L$ is a toppling matrix, then $L^{-1}$ has nonnegative entries.  Applying $L^{-1}$ to \eqref{ellwhy} yields $\yy = (1-\lambda) L^{-1} \xx$. Since $L^{-1}\xx > \zero$ and $\yy > \zero$, we conclude that $\lambda<1$.
\end{proof}

If $\Net$ is any finite abelian network and $\qq \in Q$, let $L_\qq = (I-P_\qq)D_\qq$ be the Laplacian of the local component $\Net_\qq$. Using Lemma~\ref{l.mfrompositive}, we can rephrase Theorem~\ref{t.halting} in terms of the Laplacian as follows.  The case of a toppling network is due to Gabrielov \cite{Gab94}.

\begin{corollary}
\moniker{Halting Criterion 2}
\label{c.positivedefinite}
A finite abelian network $\Net$ halts on all inputs to initial state $\qq$ if and only if $L_\qq$ is a toppling matrix.
\end{corollary}

%\begin{example}\label{ex.laplacianexample} 
\paragraph{\textbf{Example.}}
Let $\Net$ be a toppling network with three vertices $a,b,c$ and thresholds $r_a=3,\; r_b=4,\; r_c=5$. For the messages passed we have:
\begin{align*}
a \mbox{ topples } &\rightarrow \quad b \mbox{ receives } 2 \mbox{ chips, } c \mbox{ receives } 2 \mbox{ chips.} \\
b \mbox{ topples } &\rightarrow \quad a \mbox{ receives } 1 \mbox{ chip,\ghost{s} } c \mbox{ receives } 2 \mbox{ chips.} \\
c \mbox{ topples } &\rightarrow  \quad a \mbox{ receives } 0 \mbox{ chips, } b \mbox{ receives } 2 \mbox{ chips.}
\end{align*}
Note that no vertex is a sink, and vertex $a$ is productive (it ``creates'' a chip each time it topples). 
The production matrix and Laplacian of this network are:
\[ P=\left( \begin{array}{ccc}
0 & \frac14 & 0 \\
\frac 2 3 & 0& \frac 2 5 \\
\frac 2 3 & \frac 1 2 & 0\end{array} \right),
\qquad
L=\left( \begin{array}{ccc}
3 & -1 & 0 \\
-2 & 4& -2 \\
-2 & -2 & 5 \end{array} \right).\]
Since all principal minors of $L$ are positive, $L$ is a toppling matrix.  Therefore $\Net$ halts on all inputs. $\qed$
\old{
The Smith Normal Form of $L$ is: 
\[ L=\left( \begin{array}{ccc}
1 & 0 & 0 \\
0 & 1& 0 \\
0 & 0 & 34 \end{array} \right)\]
$$\Crit \Net\cong \mathbb{Z}^3/\mathbb{Z}^3L \cong \Z / 34\Z.$$
}
%\end{example}
\medskip

To any locally finite, locally irreducible abelian network $\Net$ we associate a unary network $\Sa(\Net)$ called its \emph{sandpilization}.

\begin{figure}[h!]
\centering
\includegraphics[scale=.85]{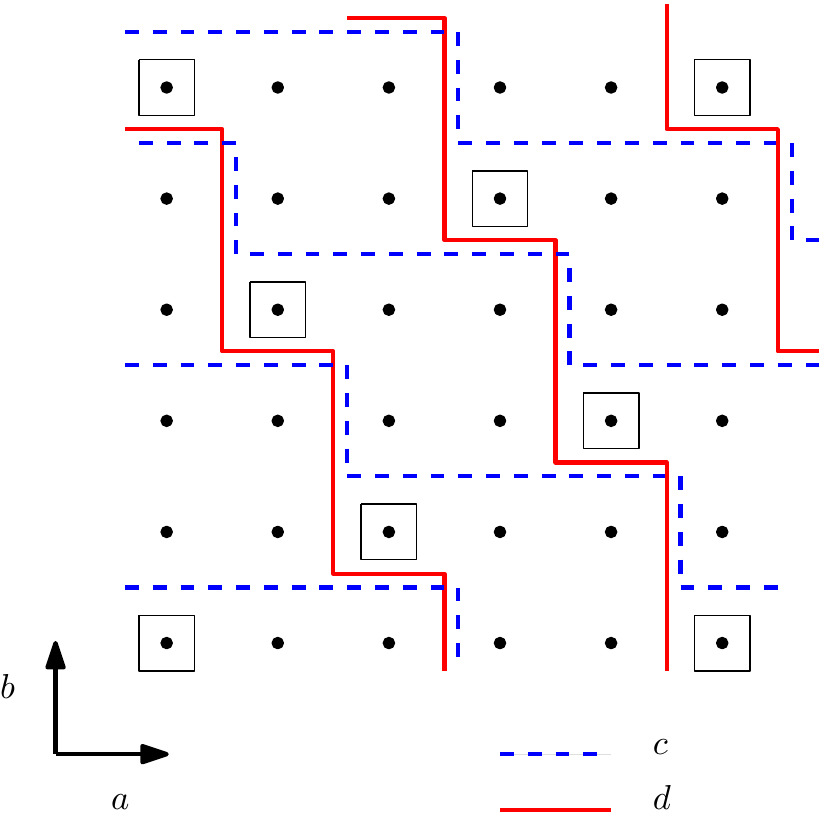}
\includegraphics[scale=.85]{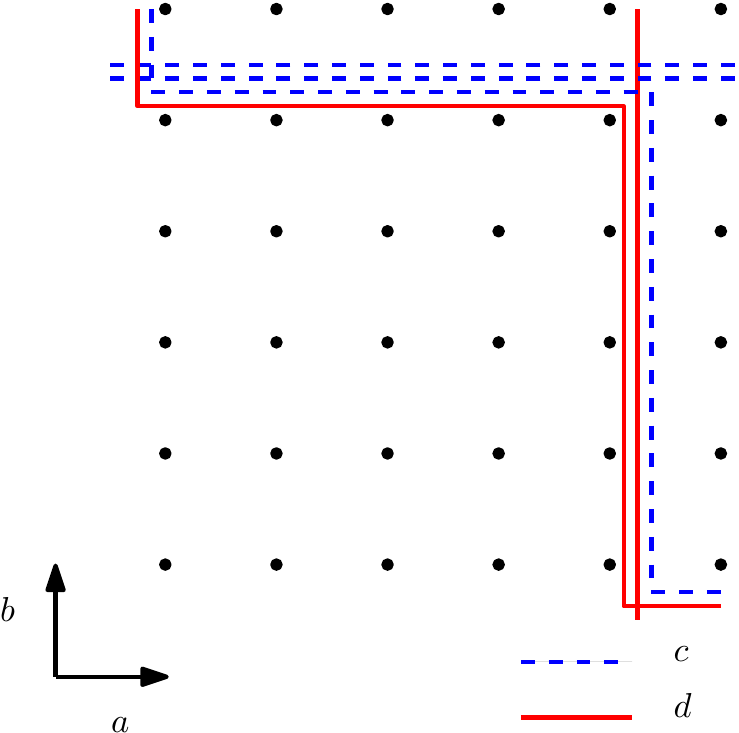}
\caption{Left: State diagram of an abelian processor with $5$ states, input alphabet $\{a,b\}$ and output alphabet $\{c,d\}$; the boxed dots all represent the same state. Right: State diagram of its sandpilization, which has $25$ states.}
\end{figure}

\begin{definition} 
(Sandpilization)
Let $\Net$ be a locally finite and locally irreducible abelian network $\Net$ with Laplacian $L$.
The \emph{sandpilization} $\Sa(\Net)$ of $\Net$ is $\Topp(L)$, the locally recurrent toppling network with Laplacian $L$.
\end{definition}

The underlying graph of $\Sa(\Net)$ is the production graph of $\Net$ (Definition~\ref{d.productiongraph}), which may be larger than the underlying graph of $\Net$.  For each $a \in A$ the processor $\Proc_a$ of $\Sa(\Net)$ has state space $Q_a = \{0,1,\ldots,r_a -1\}$ with transition $t_a(q) = q+1$ (mod $r_a$). It passes no messages except during transitions from state $r_a-1$ to state $0$, when it passes $r_a P_{ba}$ letters $b$ to each processor $\Proc_b$.  
 For example, if $\Net$ is a simple rotor network $\Rotor(G)$, then $\Sa(\Net)$ is the sandpile network $\Sand(G)$.
 
Since $\Net$ and $\Sa(\Net)$ have the same Laplacian, 
the following is immediate from Corollary~\ref{c.positivedefinite}.
 
\begin{corollary}\label{t.haltcondition}
Let $\Net$ be a finite locally irreducible abelian network.  Then $\Net$ halts on all inputs if and only if $\Sa(\Net)$ halts on all inputs. 
\end{corollary}

\subsection{Certifying that a network never halts} 

How long must we run an abelian network until we can be sure it will not halt? 
The next lemma shows that any strong amplifier gives an upper bound.  

Recall the least action principle \cite[Lemma~4.3]{part1}, which says that if $w$ is a legal execution for $\xx.\rr$ and $w'$ is a complete execution for $\xx.\rr$, then $|w| \leq |w'|$.  If $\Net$ halts on input $\xx.\rr$, then the \emph{odometer} of $\xx.\rr$ is defined as $[\xx.\rr] = |w|$, where $w$ is any complete legal execution for $\xx.\rr$.

\begin{lemma}
\label{l.nonhalting}
Let $\Net$ be a locally irreducible finite abelian network.   
Suppose that $\alpha.\qq$ is a strong amplifier for $\Net$.  If $\rr$ is locally recurrent and $\Net$ halts on input $\xx.\rr$, then any legal execution $w$ for $\xx.\rr$ satisfies $|w|_a < \alpha_a$ for some $a \in A$.
%$[\xx.\rr]_a < \alpha_a$ for some $a \in A$.
\end{lemma}

\begin{proof}
Suppose for a contradiction that $\Net$ halts on input $\xx.\rr$ but $[\xx.\rr] \geq \alpha$. Write $\uu = [\xx.\rr] = \alpha + \vv$ for some $\vv \geq \zero$. Write $\pi_{\vv}(\xx.\rr) = \yy.\ss$ where $\ss = t(\vv)\rr$. We will show that $\yy \leq \zero$, so there is a complete execution $w'$ for $\xx.\rr$ with $|w'|=\vv$. However, by the definition of the odometer there is also a legal execution $w$ for $\xx.\rr$ with $|w|=\uu$. This contradicts the least action principle since $\uu-\vv = \alpha \in \N^A - \{\zero\}$.

Since $\alpha.\qq$ is a strong amplifier we have $\alpha \acts \qq = \beta.\qq$ for some $\beta \geq \alpha$.
Now by Lemma~\ref{l.recurrent}(1) (which applies to each local action of $M_v$ on $Q_v$, as $\Net$ is 
%locally 
finite and locally irreducible) since $\rr$ is locally recurrent we have $\rr = \mm\qq$ for some $\mm \in \prod M_v$.  So $\ss = (t(\vv)\mm)\qq$, and by Lemma~\ref{l.manyamplifiers} it follows that $\alpha \acts \ss = \beta.\ss$. By the definition of the odometer, no messages remain after executing $w$, so $\pi_\uu(\xx.\rr) = \zero.\ss'$ for some state $\ss'$.
Using $\pi_{\uu} = \pi_\alpha \circ \pi_\vv$, we have
		\[ \zero.\ss' %= \pi_\uu(\xx.\rr) 
		= \pi_\alpha(\yy.\ss) = (\yy+\beta-\alpha).\ss \]
where the last equality uses Lemma~\ref{l.piprops}(ii).
Hence $\yy = \alpha - \beta \leq \zero$ as desired.
\end{proof} 
% Note: we need only assume $\Net$ is locally finite, but the conclusion is vacuous for infinite $\Net$ because our definition of halting is that there exists a *finite* complete execution.
% The proof also shows: if $\Net$ is locally finite and halts locally on input $\xx.\rr$, then the same conclusion holds.  However, for this the LAP needs to be stated and proved for infinite executions.

In the special case of a locally recurrent toppling network $\Topp(L)$, if $\yy \in \N^V - \{\zero\}$ and $L\yy \leq \zero$ then $(D\yy).\zero$ is a strong amplifier. In this case Lemma~\ref{l.nonhalting} implies that if for a particular input $\xx.\rr$ there is a legal execution in which each vertex $v$ processes at least $L_{vv} \yy_v$ letters, then $\Topp(L)$ does not halt on input $\xx.\rr$. In particular, we recover the criterion of Bj\"orner and Lov\'{a}sz \cite[Prop.\ 4.4]{BL92}: for the sandpile network $\Sand(G)$ on a directed graph $G$, if $\yy\in \N^V-\{\zero\}$ and $L\yy=\zero$ and each vertex $v$ topples at least $\yy_v$ times, then toppling persists forever. When $G$ is undirected (or Eulerian directed) we can take $\yy=\one$ and we recover the criterion of Tardos \cite[Lemma 4]{Tar88}: if each vertex topples at least once, then toppling persists forever.  

\section{Concluding Remarks}
\label{s.concluding}

We indicate here a few directions for further research on abelian networks.

\silentsubsec{Halting problem for a given input}
Theorem~\ref{t.halting} gives a polynomial time algorithm to check whether a finite abelian network $\Net$ halts on all inputs. Under what conditions is there an efficient algorithm to check whether $\Net$ halts on a \emph{given} input $\xx_0.\qq_0$?
% In the $\NP$-completeness result of \cite{BL15} for the halting problem for chip-firing, the size of a directed multigraph is measured by the number of bits needed to write down its adjacency matrix. This does not rule out the possibility that there is a polynomial time algorithm to decide halting of \Net on a given input, if we measure the size of $\Net$ by the number of bits needed to write down all of its transition functions $T_v$ and message passing functions $T_{(v,u)}$.

An inefficient algorithm runs as follows.  Let $\xx_n.\qq_n = \xx_{n-1} \acts \qq_{n-1}$ for $n \geq 1$.  By Dickson's Lemma~\ref{l.dickson} there exist $m<n$ such that $\qq_m=\qq_n$ and $\xx_m \leq \xx_n$. Each time we generate a new state $\xx_n.\qq_n$, exhaustively check for such an $m<n$. When we find one, if $\xx_n = \zero$ then $\Net$ has already halted, and if $\xx_n \neq \zero$ then $\Net$ will never halt. 

Bounds obtained from Dickson's Lemma grow very quickly \cite{FFSS11}. Lemma~\ref{l.nonhalting} suggests a possibly more efficient approach.
Given a strong amplifier $\alpha.\qq$, what is a bound for the time it takes for $\Net$ either to halt or be certified by Lemma~\ref{l.nonhalting} to run forever?  With such a bound in hand, there remains the question of which abelian networks have small (i.e. polynomial in the size of description of $\Net$) strong amplifiers.
 
 \silentsubsec{Infinite abelian networks}
Questions about the recurrence or transience of rotor walk \cite{LL09,AH11a,AH11b,FGLP14,FLP14} and the explosiveness of sandpiles \cite{FLP10} are cases of the halting problem for spatially infinite abelian networks.  In this setting, ``halting'' means that \emph{each processor} processes only finitely many letters, even though the total number of letters processed may be infinite. Generalizing the least action principle to infinite executions, along the lines of \cite{FMR09}, may be useful in approaching these and related questions.  

Among the many hard questions in this area, let us single out one.  Suppose $\Net$ is a toppling network
% or unary network
whose underlying graph is the square grid $\Z^2$.  Let $\qq$ be an initial state, and suppose that $\Net$ and $\qq$ are periodic in the sense that there is a full rank sublattice $\Lambda \subset \Z^2$ such that $\Proc_v$ and $\qq_v$ depend only on $v+\Lambda$.  Given the finite data of $\Proc_v$ and $\qq_v$ on a fundamental domain, is it decidable whether there exists an input $\xx$ with finite support such that $\Net$ does not halt on $\xx.\qq$? Cairns \cite{Hannah}, using sandpile circuits designed by Moore and Nilsson \cite{MN99}, has shown that the analogous problem in $\Z^3$ is undecidable. 
%Her construction uses only a slab $\Z^2 \times [0,N]$ of finite thickness. 
%By considering each column $\{(x,y)\} \times [0,N]$ as a single abelian processor, it follows that the $\Z^2$ question is undecidable if we allow \Net to be a locally finite abelian network.

\silentsubsec{Homotopy via embedding}
Does Lemma~\ref{l.subcomponents} have a converse? Given irreducible, strongly connected abelian networks $\Net_1$ and $\Net_2$ with $\Net_1 \approx \Net_2$ is there an irreducible $\Net$ such that $\Net_1$ and $\Net_2$ are local components of a strong component of $\Net$?

\silentsubsec{Abelian networks with coefficients}

We can define an abelian network purely in terms of monoids and without any reference to automata.  The free commutative monoid $\N^A$ played an important role in our theory. In particular, we used heavily the fact that $\N^A$ is cancellative (for instance in the proof of Lemma~\ref{l.indepofq}). What happens if we replace $\N^A$ by a different monoid?  

To make this question more precise, suppose $M$ and $M'$ are commutative monoids, written additively.  Define an \textbf{action of $M$ on $Q$ metered by $M'$} as a monoid action
	$ \nu: M \times Q \to Q $
together with a map
	\[ \mu: M \times Q \to M' \]
satisfying $\mu(\zero,\qq) = \zero'$ and
	\begin{equation} \label{e.metered} \mu(\xx+\yy, \qq) = \mu (\xx, \nu(\yy, \qq)) + \mu(\yy,\qq) \end{equation} 
for all $\xx,\yy \in M$ and all $\qq \in Q$.
% In other notation: $[(m_1+m_2).q] = [m_1.m_2q] + [m_2.q]$.
The interpretation is that $\mu(\xx,\qq)$ measures the ``cost'' (or ``byproduct'') of $\xx$ acting on $\qq$, and that costs are additive.

An example of a metered action is the local action $\acts$ of an abelian network (\textsection\ref{s.localaction}): we take $M=M'=\N^A$ with $\mu$ and $\nu$ defined by 
	\[ \xx \acts \qq = \mu(\xx,\qq) . \nu(\xx,\qq). \]
The byproduct of $\xx$ acting on $\qq$ is that some messages are passed, namely $\mu(\xx,\qq)$.

Let us define an \emph{abstract abelian network} on a directed graph $G=(V,E)$ as a collection of $4$-tuples $(M_v,Q_v,\mu_v,\nu_v)$ indexed by $v \in V$, such that $M_v$ is a commutative monoid, $Q_v$ is a set, and $(\mu_v, \nu_v)$ is an action of $M_v$ on $Q_v$ metered by
	\[ \prod_{(v,u) \in E} M_u. \]
As a special case, fix a commutative monoid $C$ and an alphabet $A = \sqcup_{v \in V} A_v$.  An \emph{abelian network with coefficients in $C$} 
%(or \emph{$C$-network}, for short) 
is an abstract abelian network with $M_v = C^{A_v}$ for all $v \in V$.  

It would be interesting to compare the computational power of such networks for different monoids $C$.  For example, taking $C=\Z$ we obtain the class of locally recurrent abelian networks (i.e., those satisfying $Q_v = e_vQ_v$ for all $v \in V$).  Taking $C=\R_+$ gives a class of networks with continuous input, which includes the abelian avalanche model of \cite{Gab93} and the divisible sandpile of \cite{LP09}. The latter computes the linear program relaxation of the integer programs that sandpiles compute \cite[Remark 4.9]{part1}.
Other $\R_+$-networks (analogous to the oil-and-water model of \cite{part1})
should compute the linear programs of \cite{Tse90}. What about $C=\Z/p\Z$?

Let us point out that the definition \eqref{e.metered} of a metered action makes sense for arbitrary monoids $M$ and $M'$, which allows us to define networks with coefficients in an arbitrary monoid $C$. Are there interesting examples with $C$ noncommutative?

\silentsec{Acknowledgments}

This research was supported by an NSF postdoctoral fellowship and NSF grants DMS-1105960 and DMS-1243606, and by the UROP and SPUR programs at MIT.


\begin{thebibliography}{9}
\bibitem[Cai15]{Hannah} Hannah Cairns, Some halting problems for abelian sandpiles are undecidable in dimension three. \arxiv{1508.00161}
\end{thebibliography}
\end{document}